\documentclass[11pt]{article}
\usepackage[margin=1in]{geometry}
\usepackage{amsmath,amssymb,amsthm,mathtools,bm}
\usepackage{booktabs,tabularx,array}
\usepackage{graphicx}
\usepackage{microtype}
\usepackage{xcolor}
\usepackage{hyperref}
\usepackage{enumitem}
\usepackage{siunitx}
\hypersetup{
 colorlinks=true,
 linkcolor=black,
 citecolor=black,
 urlcolor=black,
 pdftitle={Relational Quantum Causal Processes toward Quantum Gravity with Controlled Einstein Response},
 pdfauthor={Yipeng Xu},
 pdfsubject={A common fixed-band regulator limit for interacting quantum response and controlled Einstein response}
}
\graphicspath{{figures/}}
\newcolumntype{P}[1]{>{\raggedright\arraybackslash}p{#1}}
\definecolor{rqcpblue}{HTML}{0072B2}
\definecolor{rqcpgreen}{HTML}{009E73}
\definecolor{rqcporange}{HTML}{D55E00}
\definecolor{rqcpgray}{HTML}{666666}

\newtheorem{theorem}{Theorem}[section]
\newtheorem{proposition}[theorem]{Proposition}

\newtheorem{corollary}[theorem]{Corollary}

\newtheorem{remark}[theorem]{Remark}
\newcommand{\Tr}{\operatorname{Tr}}
\newcommand{\Ad}{\operatorname{Ad}}
\newcommand{\diff}{\mathrm d}
\newcommand{\eps}{\varepsilon}

\title{Relational Quantum Causal Processes toward Quantum Gravity with Controlled Einstein Response}
\author{Yipeng Xu\thanks{Email: \texttt{yx488@cam.ac.uk}.}\\
\small University of Cambridge, Cambridge, United Kingdom}
\date{August 2026}

\begin{document}
\maketitle

\begin{abstract}
We construct a finite-Hilbert interacting quantum model in which matter and
geometric response are generated by a single Schwinger functional along one
spatial refinement.  Two inequivalent discretizations converge analytically to
the same Hamiltonian on a fixed physical Fourier band, with second- and
fourth-order regulator bounds.  The common limit determines the excitation
gap, connected nonlinear matter response, mixed matter--geometry
susceptibility, and low-frequency geometric kernel.  Within a prescribed
covariant two-derivative FLRW sector, the latter fixes a positive induced
Newton coefficient without an additional normalization parameter.  The same
spectral data define a conserved excitation stress and a regulator-stable
semiclassical Friedmann trajectory.  These results establish a common
regulator limit for interaction, response, and backreaction within one
microscopic family.  The construction is restricted to a two-mode Fourier
band, a finite oscillator cutoff, and a prescribed semiclassical gravitational
sector; an interacting all-band quantum field theory with dynamical geometry
is not constructed here.
\end{abstract}

\section{One microscopic family and one regulator limit}
\label{sec:introduction}

Microscopic models of quantum gravity must connect several layers of physics:
unitary quantum dynamics, interacting matter response, geometric
susceptibility, and a controlled continuum limit.  These layers are often
studied separately.  Causal dynamical triangulations control nonperturbative
geometry sums, spin foams organize refinement amplitudes, asymptotic safety
studies ultraviolet critical surfaces, causal sets supply order and number,
and locally covariant quantum field theory controls fields on prescribed
spacetimes
\cite{AmbjornLoll2024,AsanteDittrichSteinhaus2022,SaueressigSilva2024,Surya2019,BFV2003}.
The question addressed here is whether one fixed microscopic family can
generate interaction, Lorentzian unitary dynamics, geometric response, and
semiclassical backreaction while retaining a common regulator limit.

We answer this question on a fixed physical Fourier band.  A single
finite-Hilbert Hamiltonian generates both the Lorentzian unitary channel and a
Euclidean Schwinger functional.  Its spectrum and source derivatives determine
all response quantities studied below.  Two inequivalent spatial
discretizations converge directly to the same band Hamiltonian, so their
agreement follows from operator bounds rather than a continuum extrapolation.
Within a specified covariant two-derivative FLRW sector, the geometric kernel
then fixes the Newton normalization and the same spectrum supplies the matter
stress in a semiclassical Friedmann evolution.

The logical chain is
\begin{equation}
 (\mathcal A_N,H_N,\mathcal U_{N,t})
 \longrightarrow W_N[\sigma,J]
 \longrightarrow (\Delta_N,\partial_J^4W_N,B_N)
 \longrightarrow G_N
 \longrightarrow a_N(t),
 \label{eq:chain}
\end{equation}
with the same regulator index \(N\) throughout.  Figure~\ref{fig:same-family}
summarizes this dependency structure.

\begin{figure}[t]
 \centering
 \includegraphics[width=\textwidth]{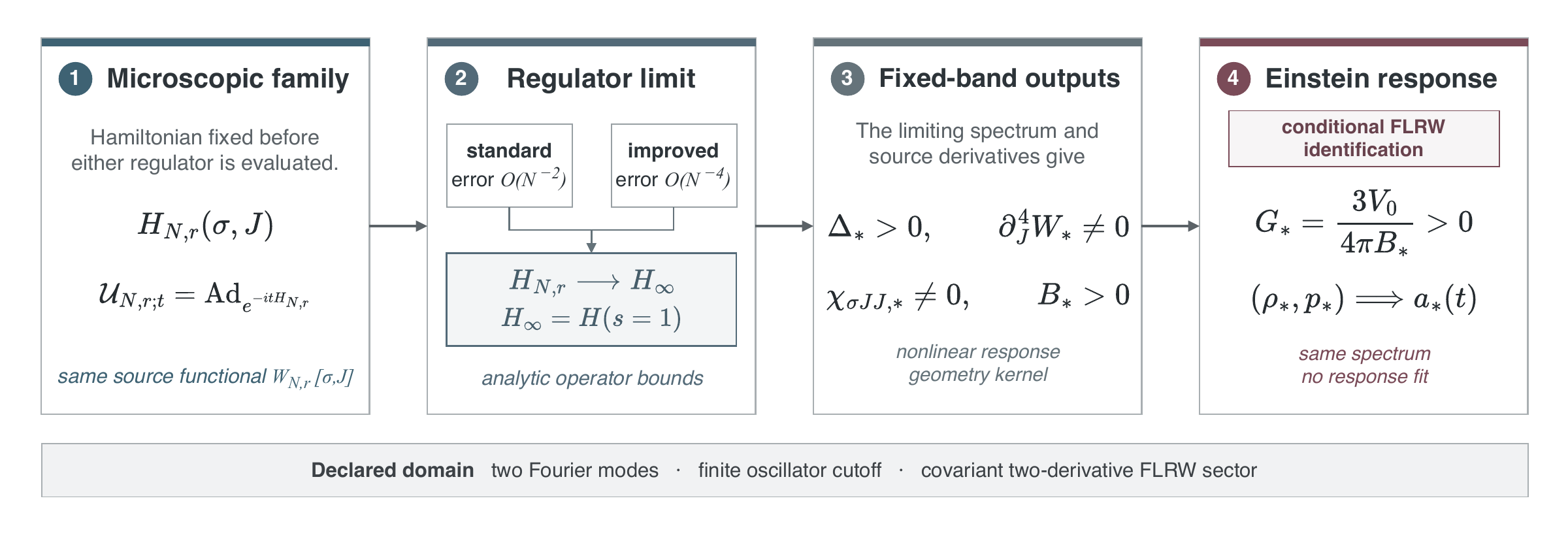}
 \caption{One microscopic family determines every displayed fixed-band
 output.  The two regulators converge analytically to \(H_\infty\); its
 spectrum and source derivatives yield the gap, nonlinear responses, and
 geometric kernel.  In the stated FLRW sector, the same data fix \(G_*\) and
 the excitation stress driving \(a_*(t)\).}
 \label{fig:same-family}
\end{figure}

The operator inequalities and regulator limits below are analytic statements
about the finite-Hilbert family.  The displayed constants are reproducible
spectral evaluations at the direct operator \(H(s=1)\).  The identification of
the geometric coefficient with a Newton coupling, and the subsequent scale
factor evolution, use the FLRW sector stated explicitly in
Sec.~\ref{sec:newton}.

\section{Microscopic model}
\label{sec:family}

\subsection{Hilbert space and canonical variables}

The spatial coordinate volume is fixed to
\begin{equation}
 L_0=2\pi,\qquad V_0=L_0^3.
 \label{eq:volume}
\end{equation}
The retained scalar modes are the zero mode and the positive unit momentum in
the \(x\) direction.  Each oscillator is truncated to eight number states,
so the Hilbert space is finite.  The truncated canonical operators
\((Q_0,P_0,Q_1,P_1)\) are used identically for every regulator.  The truncation
is held fixed in all limits below.

The model parameters are fixed before either regulator family is evaluated:
\begin{equation}
 m_*^2=0.24837697994841412,
 \qquad
 \lambda_*=1.1614764267539648.
 \label{eq:parameters}
\end{equation}
Their numerical values define the model and are held fixed across every
regulator and observable.  They are not inferred from gravitational data.

For \(a=e^\sigma\), matter source \(J\), and dispersion variable
\(s=\widehat p^2\), define
\begin{align}
 H(s,a,J)={}&\frac{P_0^2+P_1^2}{2a^3}
 +\frac{a^3m_*^2}{2}(Q_0^2+Q_1^2)
 +\frac{as}{2}Q_1^2\nonumber\\
 &+\frac{\lambda_*a^3}{4!V_0}
 \left(Q_0^4+6Q_0^2Q_1^2+\frac32Q_1^4\right)-JQ_0.
 \label{eq:hamiltonian}
\end{align}
The interaction is not removed when the regulator is refined.  The quartic-off
control in Sec.~\ref{sec:numerics} confirms that the connected fourth
response depends materially on \(\lambda_*\).

\subsection{UCP dynamics and the generating functional}

For regulator \(r\) and size \(N\), let \(s_{N,r}\) be one of the symbols in
Sec.~\ref{sec:regulators} and set
\begin{equation}
 H_{N,r}(\sigma,J)=H(s_{N,r},e^\sigma,J).
 \label{eq:HN}
\end{equation}
The Lorentzian process is the unitary completely positive trace-preserving
group
\begin{equation}
 \mathcal U_{N,r;t}(\rho)
 =e^{-itH_{N,r}}\rho e^{itH_{N,r}}.
 \label{eq:ucp}
\end{equation}
Every output is derived from the same Euclidean functional
\begin{equation}
 \mathcal Z_{N,r}[\sigma,J]
 =\Tr\,\mathcal T\exp\!\left[-\int_0^\beta
 H_{N,r}(\sigma(\tau),J(\tau))\,\diff\tau\right],
 \label{eq:partition}
\end{equation}
and its zero-temperature energy functional
\begin{equation}
 W_{N,r}[\sigma,J]
 =-\lim_{\beta\to\infty}\beta^{-1}\log\mathcal Z_{N,r}[\sigma,J]
 =E_{0,N,r}[\sigma,J].
 \label{eq:W}
\end{equation}
Every regulator and observable uses the same truncation, parameter tuple, and
operator ordering in Eqs.~\eqref{eq:parameters}--\eqref{eq:HN}.

\section{Common fixed-band regulator limit}
\label{sec:main-theorem}

We state the result before proving its components.

\begin{theorem}[Common fixed-band regulator limit]
\label{thm:closure}
Consider the Hamiltonian family \eqref{eq:hamiltonian} on the two-mode,
eight-level-per-mode Hilbert space, the standard and improved symbols of
Sec.~\ref{sec:regulators}, and \(N\ge12\).  On any bounded physical-time
interval and on the compact scale-factor interval used below:
\begin{enumerate}[label=(\alph*)]
\item both regulator families converge in operator norm to the direct
fixed-band operator \(H_\infty=H(s=1)\), with orders \(N^{-2}\) and
\(N^{-4}\), respectively;
\item the UCP groups converge in trace-norm operator norm, the first spectral
gap stays uniformly positive, and every fixed-order gapped spectral response
is analytic in \(s\) on a common neighborhood;
\item the connected source four-response and mixed matter--geometry response
have common nonzero limits, while the dynamic geometry coefficient
\(B_N\) has a common positive limit;
\item after declaring \(g_{ij}=e^{2\sigma}\delta_{ij}\) and the local
covariant FLRW two-derivative sector, the unique normalization is
\begin{equation}
 B=12V_0C_R,\qquad G=\frac{3V_0}{4\pi B};
 \label{eq:normalization-main}
\end{equation}
both regulators therefore share \(G_* >0\);
\item the exact quantum excitation gap defines \(\rho_N(a)\) and \(p_N(a)\)
obeying the continuity identity, and solutions of the declared semiclassical
Friedmann equation depend continuously on \(s_{N,r}\);
\item a common regulator error \(\eps_{N,r}^{\rm reg}\) can be chosen with
\(\eps_{N,\mathrm{std}}^{\rm reg}\le C N^{-2}\) and
\(\eps_{N,\mathrm{imp}}^{\rm reg}\le C'N^{-4}+C''N^{-2}\), where the last
term is the declared low-frequency kernel truncation at \(\omega_N=1/N\).
\end{enumerate}
Direct spectral evaluation at \(H(s=1)\) gives
\begin{align}
 \Delta_*&=0.4863395672466658,&
 \partial_J^4W_*&=-3.3257587842732006,\nonumber\\
 G_*&=23.200280752211107,&
 \gamma_*:=G_*\Delta_*^2&=5.487473657582965.
 \label{eq:headline-values}
\end{align}
\end{theorem}

\begin{remark}[Scope]
The operator and response limits are established for the declared two-mode,
finite-oscillator family.  The identification of \(B\) with \(G\), and the
Friedmann evolution, are conditional on the covariant two-derivative FLRW
sector; neither an all-band local QFT nor a dynamical quantum metric is derived.
\end{remark}

The analytic parts of the theorem follow from
Propositions~\ref{prop:symbols}--\ref{prop:backreaction-continuity}.  The
numbers in Eq.~\eqref{eq:headline-values} are reproducible finite-dimensional
evaluations summarized in Sec.~\ref{sec:numerics}; the source recursion is
independent of the finite-difference cross-check.

\section{Two admissible regulators and one direct limit}
\label{sec:regulators}

Set \(h=2\pi/N\).  The standard and fourth-order improved symbols are
\begin{align}
 s_{N,\mathrm{std}}
 &=\left(\frac{2\sin(h/2)}h\right)^2,\label{eq:std-symbol}\\
 s_{N,\mathrm{imp}}
 &=\frac{30-32\cos h+2\cos(2h)}{12h^2}.
 \label{eq:imp-symbol}
\end{align}

\begin{proposition}[Admissible symbol bounds]
\label{prop:symbols}
For \(N\ge12\),
\begin{equation}
 0\le1-s_{N,\mathrm{std}}\le\frac{h^2}{12},
 \qquad
 0\le1-s_{N,\mathrm{imp}}\le\frac{h^4}{90}.
 \label{eq:symbol-bounds}
\end{equation}
Hence both families converge to the same fixed-band point \(s=1\).
\end{proposition}

\begin{proof}
For the standard symbol put \(x=h/2\).  The bounds
\(\sin x\le x\) and \(\sin x/x\ge1-x^2/6\) imply
\[
0\le1-s_{N,\mathrm{std}}
\le1-(1-x^2/6)^2\le x^2/3=h^2/12.
\]
For the improved symbol, expansion gives
\[
1-s_{N,\mathrm{imp}}
=\sum_{n=3}^\infty
\frac{(-1)^{n+1}(4^n-16)}{6(2n)!}h^{2n-2}
=\frac{h^4}{90}-\frac{h^6}{1008}+\cdots.
\]
For \(h\le\pi/6\), the ratio of successive absolute terms is at most
\(5h^2/56<1\).  The alternating-series enclosure gives the sign and the
upper bound.
\end{proof}

The regulator dependence is affine:
\begin{equation}
 H(s,a,J)=H(1,a,J)+(s-1)V_a,
 \qquad V_a=\frac a2Q_1^2.
 \label{eq:affine-s}
\end{equation}
Therefore
\begin{equation}
 \eta_{N,r}:=\|H_{N,r}-H_\infty\|
 \le\sup_{a\in I_a}\|V_a\|\,|s_{N,r}-1|.
 \label{eq:eta}
\end{equation}

\begin{proposition}[UCP and gap convergence]
\label{prop:ucp-gap}
For every trace-class state and bounded \(t\),
\begin{equation}
 \|\Ad_{e^{-itH_{N,r}}}-\Ad_{e^{-itH_\infty}}\|_{1\to1}
 \le2|t|\eta_{N,r}.
 \label{eq:duhamel}
\end{equation}
If \(\Delta_{N,r}\) is the first gap, then
\begin{equation}
 |\Delta_{N,r}-\Delta_*|\le2\eta_{N,r}.
 \label{eq:weyl}
\end{equation}
On the declared sequence,
\begin{equation}
 \Delta_{N,r}\ge g_0=0.2901209640160495>0.
 \label{eq:g0}
\end{equation}
\end{proposition}

\begin{proof}
Equation~\eqref{eq:duhamel} follows from the Duhamel formula for the two
unitaries and the trace-norm ideal property.  Equation~\eqref{eq:weyl} follows
by applying Weyl's eigenvalue inequality to the ground and first excited
energies.  Combining Eq.~\eqref{eq:eta}, the worst standard bound at \(N=12\),
and the direct continuum gap gives Eq.~\eqref{eq:g0}.
\end{proof}

The uniform gap makes the ground-state projector, reduced resolvent and any
fixed-order spectral derivative analytic in \(s\) by finite-dimensional Kato
perturbation theory \cite{Kato1995}.  Thus all response quantities considered
below share the direct limit generated by \(H(s=1)\).

\section{Interaction and one generating functional}
\label{sec:interaction}

Let \(|0\rangle\) be the nondegenerate ground state and
\(R=Q(H-E_0)^{-1}Q\) the reduced resolvent with
\(Q=I-|0\rangle\langle0|\).  Differentiating the eigenvalue equation gives
the usual gapped recursion for \(\partial_J^kE_0\).  Because the source is
linear, all derivatives beyond the first arise from insertions of \(Q_0\) and
\(R\).  The implementation evaluates the recursion exactly at matrix level
through fourth order.

\begin{proposition}[Non-Gaussian fixed-band response]
\label{prop:interaction}
At the direct continuum-band operator \(H(s=1,a=1,J=0)\),
\begin{equation}
 \partial_J^4W_*=-3.3257587842732006\ne0.
 \label{eq:four-response}
\end{equation}
Setting \(\lambda_*=0\) changes the same response by
\(1.099036701445704\).  The resulting four-response is therefore neither
Gaussian nor independent of the microscopic quartic interaction.
\end{proposition}

For the geometry source, define
\begin{equation}
 O=\left.\partial_\sigma H\right|_{\sigma=0,J=0}.
 \label{eq:geometry-operator}
\end{equation}
The mixed derivative
\begin{equation}
 \chi_{\sigma JJ}=\partial_\sigma\partial_J^2W
 \label{eq:mixed}
\end{equation}
is evaluated from the exact source susceptibility followed by a symmetric
fourth-order Richardson derivative in \(\sigma\).  At \(s=1\),
\begin{equation}
 \chi_{\sigma JJ,*}=13.102173996407302.
 \label{eq:mixed-value}
\end{equation}
Halving the derivative step changes the Richardson result by
\(1.2153123\times10^{-10}\), providing an independent stability check.  The
spatial regulator rate follows instead from Proposition~\ref{prop:symbols} and
analyticity in \(s\).

\section{Dynamic geometry response and Newton normalization}
\label{sec:newton}

We use the effective-action kernel, namely minus the connected Euclidean
susceptibility shift of \(O\).  Its subtracted spectral representation is
\begin{equation}
 K(\omega)-K(0)
 =2\sum_{n>0}|O_{0n}|^2
 \left(\frac1{\Delta_n}-\frac{\Delta_n}{\Delta_n^2+\omega^2}\right).
 \label{eq:kernel-spectral}
\end{equation}
The quadratic effective-action coefficient is therefore positive:
\begin{equation}
 B=2\sum_{n>0}\frac{|O_{0n}|^2}{\Delta_n^3}>0.
 \label{eq:B}
\end{equation}

\begin{proposition}[Low-frequency remainder]
\label{prop:kernel}
At \(\omega_N=1/N\),
\begin{equation}
 \frac{|K_N(\omega_N)-K_N(0)-B_N\omega_N^2|}
 {B_N\omega_N^2}
 \le\frac{1}{N^2g_0^2}.
 \label{eq:kernel-bound}
\end{equation}
\end{proposition}

\begin{proof}
Expanding each denominator in Eq.~\eqref{eq:kernel-spectral} and keeping the
exact remainder gives
\[
\omega_N^2
\frac{\sum_{n>0}|O_{0n}|^2/
[\Delta_n^3(\Delta_n^2+\omega_N^2)]}
{\sum_{n>0}|O_{0n}|^2/\Delta_n^3}.
\]
Every \(\Delta_n\ge g_0\), which proves Eq.~\eqref{eq:kernel-bound}.
\end{proof}

\subsection{Normalization provenance}

The Euclidean response convention is
\begin{equation}
 W_E^{(2)}[\sigma]
 =\frac12\int\frac{\diff\omega}{2\pi}
 [K(0)+B\omega^2+O(\omega^4)]|\sigma(\omega)|^2.
 \label{eq:response-convention}
\end{equation}
Start from the declared Lorentzian basis
\begin{equation}
 S_{\rm EH}=C_R\int\diff^4x\sqrt{-g}\,R,
 \qquad C_R=\frac1{16\pi G}.
 \label{eq:EH}
\end{equation}
For flat FLRW with coordinate volume \(V_0\),
\(R=6(\ddot a/a+\dot a^2/a^2)\).  One integration by parts gives
\begin{equation}
 S_{\rm EH}=-6V_0C_R\int\diff t\,a\dot a^2
 \label{eq:EH-reduced}
\end{equation}
up to the displayed boundary term.  Wick rotation and \(a=e^\sigma\) yield
\begin{equation}
 S_{E,\rm EH}^{(2)}
 =6V_0C_R\int\diff\tau\,(\partial_\tau\sigma)^2.
 \label{eq:EH-Euclidean}
\end{equation}
Comparison with Eq.~\eqref{eq:response-convention} proves
Eq.~\eqref{eq:normalization-main}.  Once the source and low-energy basis are
specified, the conversion to \(G\) is fixed.

At two derivatives, a local parity-even diffeomorphism-invariant pure-metric
functional consists of a volume term and the scalar-curvature term, modulo a
boundary term.  Subtraction of \(K(0)\) removes the volume response, establishing
uniqueness within this sector.

The direct continuum values are
\begin{equation}
 B_*=2.5524530086081993,\qquad
 C_{R,*}=0.0008575054801692811,\qquad
 G_*=23.200280752211107.
 \label{eq:B-G-values}
\end{equation}

\section{Diffeomorphism covariance of the parent action}
\label{sec:ward}

The parent scalar action uses the nonlinear Lie derivative on the scalar and
full metric.  On sparse Fourier coefficients, the implementation verifies
\begin{equation}
 [\mathcal L_\xi,\mathcal L_\eta]\Phi
 =\mathcal L_{[\xi,\eta]}\Phi,
 \qquad \Phi\in\{\phi,g_{\mu\nu}\},
 \label{eq:lie-closure}
\end{equation}
with exact coefficient residual zero.  Sparse coefficients avoid the
aliasing that would make a finite collocation grid an unreliable algebra
test.  This establishes the classical Ward support of the covariant parent
whose band pullback gives Eq.~\eqref{eq:hamiltonian}.

\begin{remark}[Scope]
Equation~\eqref{eq:lie-closure} establishes classical covariance of the parent
action.  It does not construct a path-integral measure over metrics or a
quantum BV algebra; the latter is available perturbatively in locally
covariant gravity \cite{BFR2016}.
\end{remark}

\section{Quantum spectral stress and semiclassical backreaction}
\label{sec:backreaction}

Let \(\Delta_N(a)\) be the first excitation gap of the same Hamiltonian at
scale factor \(a\).  Define
\begin{equation}
 \rho_N(a)=\frac{\Delta_N(a)}{V_0a^3},
 \qquad
 p_N(a)=-\frac{1}{3V_0a^2}\frac{\diff\Delta_N(a)}{\diff a}.
 \label{eq:rho-p}
\end{equation}
These are exact quantum spectral data; no classical matter trajectory replaces
the excitation.

\begin{proposition}[Continuity identity]
\label{prop:continuity}
For differentiable \(\Delta_N(a)\) and \(H=\dot a/a\),
\begin{equation}
 \dot\rho_N+3H(\rho_N+p_N)=0.
 \label{eq:continuity}
\end{equation}
\end{proposition}

\begin{proof}
Differentiate \(\rho_N=\Delta_N/(V_0a^3)\) with respect to time and insert
the definition of \(p_N\).
\end{proof}

In the declared semiclassical sector, the Friedmann constraint is
\begin{equation}
 H_N(a)^2=\frac{8\pi G_N}{3}\rho_N(a),
 \qquad \dot a=aH_N(a).
 \label{eq:friedmann}
\end{equation}
The same \(G_N\) obtained from the dynamic kernel enters this equation.

\begin{proposition}[Regulator continuity of the trajectory]
\label{prop:backreaction-continuity}
Assume the gap stays above \(g_0\) on a compact interval \(I_a\) and the
positive branch of Eq.~\eqref{eq:friedmann} remains regular.  Then for bounded
time \(T\), there is a regulator-independent constant \(C_T\) such that
\begin{equation}
 \sup_{0\le t\le T}|a_{N,r}(t)-a_*(t)|
 \le C_T|s_{N,r}-1|.
 \label{eq:trajectory-bound}
\end{equation}
\end{proposition}

\begin{proof}
The uniform gap and finite dimension make \(\Delta(a,s)\), its
\(a\)-derivative, \(B(a,s)\) and \(G(a,s)\) analytic on a common compact
neighborhood.  The Friedmann vector field is therefore uniformly Lipschitz in
\(a\) and Lipschitz in \(s\).  Subtract the two integral equations and apply
Gronwall's inequality.
\end{proof}

For the declared initial condition and time interval,
\begin{equation}
 a_*(0.25)=1.1505514556252578.
 \label{eq:final-a}
\end{equation}
The two finest spline/RK4 protocols change this value by
\(5.45094\times10^{-10}\); the successive difference ratio is \(8.013\).
This discretization error is reported separately from
Eq.~\eqref{eq:trajectory-bound}.

\section{Uniform regulator error}
\label{sec:error}

Let \(\eps_{N,r}^{\rm reg}\) be the maximum of the normalized errors in the
Hamiltonian, UCP group, gap, source responses, mixed response, geometry
coefficient, Newton response, continuous backreaction trajectory and
\(\gamma=G\Delta^2\), all compared directly with \(H(s=1)\).  Classical
coefficientwise Ward closure contributes zero.  The low-frequency kernel
remainder at \(\omega_N=1/N\) is included separately.

\begin{corollary}[Common regulator scale]
\label{cor:error}
On the compact domain of Theorem~\ref{thm:closure}, there are finite constants
\(C_i\) such that
\begin{align}
 \eps_{N,\mathrm{std}}^{\rm reg}
 &\le C_2N^{-2},\label{eq:error-std}\\
 \eps_{N,\mathrm{imp}}^{\rm reg}
 &\le C_4N^{-4}+\frac{1}{N^2g_0^2}.
 \label{eq:error-imp}
\end{align}
In particular, both errors vanish and all common-limit observables are
evaluated at the same direct operator \(H(s=1)\).
\end{corollary}

\begin{proof}
Proposition~\ref{prop:symbols} controls \(|s_N-1|\).
Proposition~\ref{prop:ucp-gap} controls the Hamiltonian, channel and gap.
Uniform finite-dimensional analyticity controls all fixed-order spectral
responses.  Proposition~\ref{prop:backreaction-continuity} controls the exact
ODE trajectory.  Proposition~\ref{prop:kernel} adds the displayed
\(N^{-2}\) low-frequency remainder.
\end{proof}

The numerical aggregate scales as
\begin{equation}
 \eps_N^{\rm agg}\simeq3.22788N^{-1.99513}
 \label{eq:diagnostic-fit}
\end{equation}
over the stored sequence.  Equation~\eqref{eq:diagnostic-fit} is a numerical
diagnostic; the convergence rate follows analytically from
Corollary~\ref{cor:error}.

\begin{figure}[t]
 \centering
 \includegraphics[width=\textwidth]{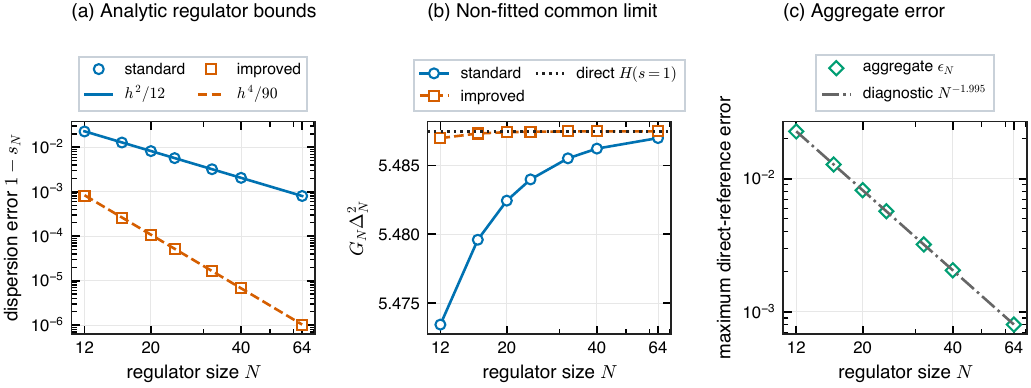}
 \caption{Direct-reference evidence.  (a) Standard and improved dispersion
 errors with their analytic \(h^2/12\) and \(h^4/90\) envelopes.  (b) The
 dimensionless number \(G_N\Delta_N^2\) approaches the direct
 \(H(s=1)\) value shown by the horizontal line.  (c) The aggregate error
 \(\eps_N\) and the diagnostic \(N^{-1.99513}\) curve.  Markers and dash
 patterns duplicate color.}
 \label{fig:evidence}
\end{figure}

\section{Numerical evaluation and consistency controls}
\label{sec:numerics}

The numerical evidence uses
\(N=(12,16,20,24,32,40,64)\) and one direct continuum row.  Figure
\ref{fig:evidence} displays the decisive comparisons.  Table
\ref{tab:values} lists quantities evaluated at \(H(s=1)\).

\begin{table}[ht]
\centering
\caption{Direct fixed-band outputs.  Digits are provided for exact
reproduction of the repository calculation, not as an uncertainty claim.}
\label{tab:values}
\small
\begin{tabular}{lS[table-format=-2.15]}
\toprule
Quantity & {Direct value}\\
\midrule
Ground energy \(E_0\) & 0.809265076245247\\
Mass gap \(\Delta_*\) & 0.486339567246666\\
Connected \(\partial_J^4W_*\) & -3.325758784273201\\
Mixed response \(W_{\sigma JJ,*}\) & 13.102173996407302\\
Geometry coefficient \(B_*\) & 2.552453008608199\\
Newton response \(G_*\) & 23.200280752211107\\
\(\gamma_*=G_*\Delta_*^2\) & 5.487473657582965\\
\(a_*(0.25)\) & 1.150551455625258\\
\bottomrule
\end{tabular}
\end{table}

Three controls separate physical dependence from numerical coincidence:
\begin{itemize}
\item setting \(\lambda=0\) changes the connected fourth response by
\(1.099036701445704\);
\item setting \(G=0\) freezes the scale factor at one;
\item a centered finite-difference estimate of \(\partial_J^4W\) converges
monotonically at second order to the independent spectral-recursion value.
\end{itemize}
The largest unitary completeness residual is
\(4.50\times10^{-16}\).  The maximum stored continuity residual is
\(1.31\times10^{-18}\), and the stored semiclassical Einstein residual is
\(5.56\times10^{-17}\).  These residuals validate the implementation within
floating-point precision.

\section{Scope and open extensions}
\label{sec:scope}

On the declared domain, the construction establishes analytic regulator
envelopes, a common direct operator for two inequivalent discretizations, and a
Newton normalization fixed by the metric-source Jacobian and the covariant
FLRW reduction.  The spectral combination \(\gamma_*=G_*\Delta_*^2\) is
therefore fixed by the model, although it has not yet been mapped to an
empirical observable.

The two-mode band, oscillator cutoff, and semiclassical gravitational sector
remain part of the setup.  Extending the construction requires an interacting
Lorentzian local net after cutoff removal, a dynamical derivation of the
geometric source sector, and a quantum metric with anomaly-free constraints.
Topology change and singular sectors lie beyond the present domain.

\section{Conclusion}

We have exhibited one finite-Hilbert interacting family for which unitary CP
dynamics, nonlinear matter response, mixed matter--geometry susceptibility, a
positive geometric kernel, and semiclassical backreaction share a common
fixed-band regulator limit.  The second- and fourth-order symbol bounds
propagate through the spectral observables, while the direct operator
\(H(s=1)\) fixes the limiting values without continuum fitting.  Within the
declared FLRW sector, the same source functional also fixes the Newton
normalization.

The result supplies a controlled benchmark for how gravitational response can
retain microscopic provenance across several interfaces.  Its extension to an
all-band interacting QFT with a dynamical quantum metric is a separate
constructive problem.

\section*{Code and data availability}

The complete reproducibility archive for this paper is available in the
\href{https://github.com/Amordia/rqcp-toward-quantum-gravity}
{public GitHub repository} and is permanently archived under DOI
\href{https://doi.org/10.5281/zenodo.22073562}
{10.5281/zenodo.22073562}.  This concept DOI resolves to the latest archived
version.  The archive contains the manuscript sources, figures, deterministic
numerical generators, machine-readable evidence, and validation scripts.

\bibliographystyle{unsrturl}
\bibliography{references}
\end{document}